\documentclass[11pt]{article}
\usepackage{graphicx}
\usepackage{tabularx}
\usepackage{url}
\usepackage{amsthm} 
\usepackage{amsmath} 
\usepackage{amsfonts}

\usepackage{color}
\usepackage{xspace}
\usepackage{tikz}
\usetikzlibrary{decorations,decorations.pathmorphing,decorations.pathreplacing,fit,positioning,arrows}

\theoremstyle{plain}
\newtheorem{theorem}{Theorem}
\newtheorem{lemma}{Lemma}

\newtheorem{claim}{Claim}

\newtheorem{corollary}{Corollary}
\theoremstyle{definition}

\newtheorem{conjecture}{Conjecture}

\theoremstyle{remark}

\date{}

\begin{document}

\title{Maximum independent sets in subcubic graphs: \\ new results}

\author{
Ararat Harutyunyan\thanks{Universit\'e Paris-Dauphine, Universit\'e PSL, CNRS, LAMSADE, 75016 PARIS, FRANCE. Email:  ararat.harutyunyan@dauphine.fr} \and
Michael Lampis\thanks{Universit\'e Paris-Dauphine, Universit\'e PSL, CNRS, LAMSADE, 75016 PARIS, FRANCE. Email:  michail.lampis@dauphine.fr} \and
Vadim Lozin\thanks{Mathematics Institute, University of Warwick, Coventry, CV4 7AL, UK. Email:  V.Lozin@warwick.ac.uk}
\and 
J{\'{e}}r{\^{o}}me Monnot\thanks{Universit\'e Paris-Dauphine, Universit\'e PSL, CNRS, LAMSADE, 75016 PARIS, FRANCE. Email:  Jerome.Monnot@dauphine.fr} }

\maketitle
\begin{abstract}
The maximum independent set problem is known to be NP-hard in the class of subcubic graphs, i.e. graphs of vertex degree at most 3.
We present a polynomial-time solution in a subclass of subcubic graphs generalizing several previously known results. 

\end{abstract}

\section{Introduction}
In a graph, an {\it independent set} is a subset of vertices no two of which are adjacent. The maximum independent set 
problem asks to find in a graph $G$ an independent set of maximum size. The size of a maximum independent set in $G$ is 
called the {\it independence number} of $G$ and is denoted $\alpha(G)$.  

The maximum independent set problem is one of the first problems that has been shown to be NP-hard. Moreover, the problem 
remains NP-hard under substantial restrictions. In particular, it is NP-hard for graphs of vertex degree at most 3, also 
known as {\it subcubic} graphs. In terms of vertex degree, this is the strongest possible restriction under which the problem 
remains NP-hard, since for graphs of vertex degree at most 2 the problem is solvable in polynomial time. However, with 
respect to other parameters the restriction to subcubic graphs is not best possible, as the problem remains NP-hard 
for subcubic graphs of girth at least $k$ for any fixed value of $k$ \cite{large-girth}, where the girth of a graph is the size 
of a smallest cycle. In other words, the problem is NP-hard  for $(C_3,\ldots,C_k)$-free subcubic 
graphs for each value of $k$, where $C_k$ is a chordless cycle of length $k$. The idea behind this conclusion is quite simple:
it is not difficult to see that  a double subdivision of an edge increases the independence number of the graph by exactly one,
and hence, by repeatedly subdividing the edges of a subcubic graph $G$ we can destroy all small cycles in $G$, i.e. we can transform 
$G$ into a graph of large girth. 

Let us observe that by means of edge subdivisions we can also destroy small copies of some other graphs, in particular, 
graphs of the form $H_k$ represented in Figure~\ref{fig:S} (left) . Therefore, the maximum independent set problem remains  
NP-hard for $(C_3,\ldots,C_k,H_1,\ldots,H_k)$-free subcubic graphs for each value of $k$. 


Let us denote by $S_k$ the class of $(C_3,\ldots,C_k,H_1,\ldots,H_k)$-free subcubic graphs and by $\kappa(G)$ the maximum $k$ such that $G\in S_k$. 
If $G$ belongs to no class $S_k$, then $\kappa(G)$ is defined to be $0$, and if $G$ belongs to all classes $S_k$, then $\kappa(G)$ is defined to be $\infty$. 
Also, for a set of graphs $M$, $\kappa(M)$ is defined as $\kappa(M) =\sup\{\kappa(G)\  :\  G \in M\}$. With this notation, we can derive the following 
conclusion from the above discussion (see e.g. \cite{large-apple}). 

\begin{theorem}
Let $M$ be a set of graphs.
If $\kappa(M)< \infty$, then the maximum independent set problem is NP-hard in the class of  $M$-free subcubic graphs.
\end{theorem}

This theorem suggests that, unless $P = NP$, the maximum independent set problem is solvable in polynomial time in the class of $M$-free graphs only if the parameter $\kappa$ 
is unbounded in the set $M$. There are three basic ways to unbind this parameter in $M$: 
\begin{itemize}
\item[1.] include in $M$ a graph $G$ with $\kappa(G) =\infty$; 
\item[2.] include in $M$ graphs with arbitrarily large induced cycles; 
\item[3.] include in $M$ graphs with arbitrarily large induced subgraphs of the form $H_k$. 
\end{itemize}

To give an example of a polynomial-time result of the first type, let us observe that $\kappa(G) =\infty$  if and only if  every connected component of $G$ has the form $S_{i,j,k}$
represented in Figure~\ref{fig:S} (right). We call any graph of the form $S_{i,j,k}$ a {\it tripod}. 
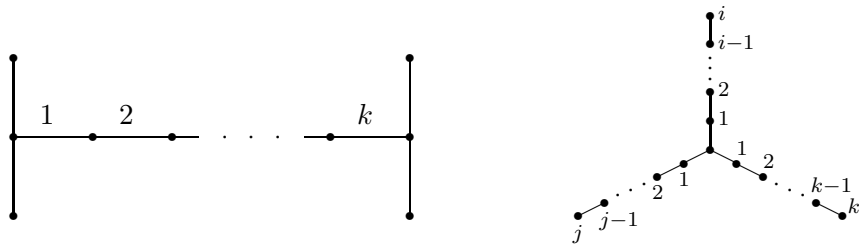
\begin{figure}[ht]
\begin{center} 
\begin{picture}(200,100)
\put(50,40){\circle*{3}} \put(80,40){\circle*{3}} \put(110,40){\circle*{3}}
\put(130,40){\circle*{1}} \put(140,40){\circle*{1}} \put(150,40){\circle*{1}}
\put(170,40){\circle*{3}} \put(200,40){\circle*{3}} \put(50,70){\circle*{3}}
\put(50,10){\circle*{3}} \put(200,70){\circle*{3}} \put(200,10){\circle*{3}}
\put(50,40){\line(1,0){30}} \put(80,40){\line(1,0){30}} \put(110,40){\line(1,0){10}}
\put(160,40){\line(1,0){10}} \put(170,40){\line(1,0){30}}
\put(50,40){\line(0,1){30}} \put(50,40){\line(0,-1){30}}
\put(200,40){\line(0,1){30}} \put(200,40){\line(0,-1){30}} \put(60,45){1}
\put(90,45){2} \put(180,45){$k$}
\end{picture}
\begin{picture}(200,100)
\put(110,35){\circle*{3}}
\put(110,46){\circle*{3}}
\put(110,57){\circle*{3}}
\put(110,75){\circle*{3}}
\put(110,86){\circle*{3}}
\put(110,62){\circle*{1}}
\put(110,66){\circle*{1}}
\put(110,70){\circle*{1}}
\put(110,35){\line(0,1){11}}
\put(110,46){\line(0,1){11}}
\put(110,75){\line(0,1){11}}
\put(100,30){\circle*{3}}
\put(90,25){\circle*{3}}
\put(70,15){\circle*{3}}
\put(60,10){\circle*{3}}
\put(85,22){\circle*{1}}
\put(80,20){\circle*{1}}
\put(75,18){\circle*{1}}
\put(110,35){\line(-2,-1){10}}
\put(100,30){\line(-2,-1){10}}
\put(70,15){\line(-2,-1){10}}
\put(120,30){\circle*{3}}
\put(130,25){\circle*{3}}
\put(150,15){\circle*{3}}
\put(160,10){\circle*{3}}
\put(135,22){\circle*{1}}
\put(140,20){\circle*{1}}
\put(145,18){\circle*{1}}
\put(110,35){\line(2,-1){10}}
\put(120,30){\line(2,-1){10}}
\put(150,15){\line(2,-1){10}}
\put(113,46){$_1$}
\put(113,57){$_2$}
\put(113,75){$_{i-1}$}
\put(113,86){$_i$}
\put(98,23){$_1$}
\put(88,18){$_2$}
\put(68,8){$_{j-1}$}
\put(58,3){$_j$}
\put(120,35){$_1$}
\put(130,30){$_2$}
\put(148,21){$_{k-1}$}
\put(162,13){$_k$}
\end{picture} 
\end{center}
\caption{The graphs $H_k$ (left) and  $S_{i,j,k}$ (right)}
\label{fig:S}
\end{figure}

In other words, if the set $M$ of forbidden induced subgraphs is finite, then $M$ must contain a graph for which every component is a tripod
for the maximum independent set problem in the class of $M$-free subcubic graphs to be polynomial-time solvable (assuming P$\neq$NP). 
In \cite{M2I}, it was conjectured that this condition is also sufficient. Moreover, for graphs of bounded vertex degree the problem can be easily reduced 
to connected forbidden induced graphs, in which case the conjecture can be restated as follows. 

\begin{conjecture}\label{con:1}
The maximum independent set problem is polynomial-time solvable for $G$-free subcubic graphs if and only if $G$ is a tripod.
\end{conjecture} 

One of the minimal non-trivial tripods is the claw $S_{1,1,1}$. The problem can be solved for the claw-free graphs in polynomial time even without the restriction to bounded degree graphs 
\cite{claw}. In \cite{fork}, the result for claw-free graphs was extended to $S_{1,1,2}$-free graphs, also known as fork-free graphs, and again without the restriction 
to bounded degree graphs. However, any further extension becomes much harder even for bounded degree graphs, and only recently a solution was found for $S_{2,2,2}$-free 
subcubic graphs \cite{subcubic}. Currently, this is  one of the few maximal subclasses of subcubic graphs with polynomial-time solvable independent set problem.

\medskip
Now we turn to polynomial-time solutions of the second type, i.e. classes of graphs where forbidden induced subgraphs contain arbitrarily large chordless cycles. 
Clearly, in this case the set of  forbidden induced subgraphs must be infinite. A typical example of this type deals with classes of bounded chordality, i.e.
classes excluding {\it all} chordless cycle of length at least $k$ for a constant $k$. Without a restriction to bounded degree graphs a solution of  this type is
known only for $k=4$, i.e. for chordal graphs \cite{chordal}, and is unknown for larger values of $k$. Together with the restriction to bounded degree graphs bounded chordality 
implies bounded tree-width \cite{chordality} and hence polynomial-time solvability of the maximum independent set problem. 
In other words, the problem can be solved for $(C_k,C_{k+1},\ldots)$-free graphs of bounded vertex degree for each value of $k\ge 3$.

An {\it apple} $A_k$, $k\ge 4$, is a graph formed of a chordless cycle $C_k$ and an additional vertex, called the {\it stem}, which has exactly one neighbour on the cycle $C_k$. 
The class of $(A_4,A_5,\ldots)$-free graphs generalizes both chordal graphs and claw-free graphs, and  a solution for the maximum independent set problem in this class was 
presented in \cite{apple}. In case of bounded degree graphs this solution can be extended to graphs without {\it large} apples, i.e. to $(A_k,A_{k+1},\ldots)$-free graphs of bounded vertex degree for any 
fixed value of $k$ \cite{large-apple}.

\medskip
Generalizing both the subcubic graphs without large apples and $S_{2,2,2}$-free subcubic graphs, in the present paper we prove polynomial-time solvability of the maximum independent set 
problem for subcubic graphs excluding large apples with a long stem. An {\it apple with a long stem} $A^*_k$ is obtained from an apple $A_k$ by adding one more vertex which is adjacent to the stem of $A_k$ only.  
We show that for any fixed value of $k$, the maximum independent set problem in the class of $(A^*_k,A^*_{k+1},\ldots)$-free subcubic graphs can be solved in polynomial time.
Observe that this class contains all $S_{2,p,p}$-free subcubic graphs for any fixed $p<k$ and hence our result brings us much closer to the proof of Conjecture~\ref{con:1}.


\section{Preliminaries}
\label{sec:pre}
All graphs in this paper are simple, i.e. indirected, without loops and multiple edges. The vertex set and the edge set of a graph $G$ is denoted by $V(G)$ and $E(G)$, respectively. 
The {\it neighbourhood} $N(v)$ of a vertex $v\in V(G)$ is the set of vertices of $G$ adjacent  to $v$.
The {\it degree} of $v\in V(G)$ is the number of its neighbours, i.e. $|N(v)|$. 
As usual, $P_n$ and $C_n$ denote a chordless path and a chordless cycle with $n$ vertices, respectively,

A subgraph of $G$ induced by a subset $U\subseteq V(G)$ is denoted $G[U]$. If $G$ contains no induced subgraphs isomorphic to a graph $H$, we say that $G$ is $H$-free.

\medskip
{\it Outline of the proof}. To prove polynomial-time solvability of the maximum independent set problem in the class of $(A^*_k,A^*_{k+1},\ldots)$-free subcubic graphs,
\begin{itemize}
\item[1.] We start by checking if the input graph $G$ has an induced copy of $S_{2,2,2}$. If $G$ is $S_{2,2,2}$-free, then the problem can be solved for $G$ in polynomial time \cite{subcubic}.
Otherwise, we proceed to checking whether $G$ has a induced cycle of length at
least $p=300k$. This can be done in polynomial time, as shown in
Lemma~\ref{lem:cycles} below. If $G$ does not contain induced cycles of length
at least $p$, then the tree-width of $G$ is bounded by a function of $k$
\cite{chordality} and hence the problem can be solved in polynomial time for
$G$.

\item[2.]
If $G$ contains an induced copy of $S_{2,2,2}$ and a large induced cycle $C$, then in the absence of large induced apples with long stems it must contain a large extended
cycle $C^*$, which is a graph obtained from $C$ by adding two vertices that create a $C_6$ together with four consecutive vertices of $C$ (see Figure~\ref{fig:extended} in Section~\ref{sec:destroy}).  
This is proved in Section~\ref{sec:to}. An important feature of this proof is the assumption that the input graph $G$ is connected and has no separating cliques, i.e. cliques whose 
removal disconnects the graph. A polynomial-time reduction of the maximum independent set problem to graphs without separating cliques can be found in \cite{clique-separator-1,clique-separator-2}.
\item[3.] Finally, in Section~\ref{sec:destroy} we destroy a large extended cycle by means of various local reductions. 
Each of them transforms $G$ into a smaller graph $G'$ in the same class with a fixed difference $\alpha(G)-\alpha(G')$.
The set of reductions is described in Section~\ref{sec:reductions} and their application to a graph $G$ containing a large extended cycle is described in Section~\ref{sec:analysis}. 
By destroying the large extended cycle $C^*$, we  destroy either the  cycle $C$ or the induced copy of $S_{2,2,2}$ (or both) and return to Step 1 to check if there are other copies of a large induced cycle or
an  induced $S_{2,2,2}$.    
\end{itemize}

Before proceeding, let us state the Lemma that we used in part 1 of the above outline.

\begin{lemma}\label{lem:cycles}

For each $p$ there is an algorithm running in time $n^{O(p)}$ which decides if a given $n$-vertex graph contains an induced cycle of length at least $p$.

\end{lemma}

\begin{proof}

We repeat the following procedure for every induced path $P$ with $p-1$
vertices. Let $u,v$ be the two endpoints of $P$. We delete from the graph all
vertices of $P$ except $u,v$ and all vertices of $V\setminus P$ that have a
neighbor in $P\setminus \{u,v\}$. If the resulting graph has a path from $u$ to
$v$ then a shortest such path together with $P$ gives an induced cycle of
length at least $p$ in $G$. Conversely, if there exists an induced cycle in $G$
that contains $P$ then there must exist a path from $u$ to $v$ in the resulting
graph. Since the total number of induced paths on $p-1$ vertices is at most
$n^{p-1}$ we get the promised running time.  \end{proof}


\section{From large cycles to extended large cycles}
\label{sec:to}
We recall that $C^*$ denotes an {\it extended cycle}, i.e. the graph obtained from a cycle $C$ by adding two vertices that create a $C_6$ together with 
four consecutive vertices of $C$ (see Figure~\ref{fig:extended} in Section~\ref{sec:destroy}). Also, $A^*_p$ denotes an {\it apple with a long stem}, where 
$p$ stands for the size of the cycle in the apple.  
An apple with a long stem consisting of a cycle $C$ and two stem vertices $x,y$ will be denoted $C_{x,y}$.

The main goal of this section is to show that if $G$ contains a large induced cycle and an induced copy of $S_{2,2,2}$, 
then it contains  either a large induced extended cycle or a large induced apple with a long stem. 
This will be shown in two steps in Lemmas~\ref{lem:b1} and~\ref{lem:b2} 

\begin{lemma}\label{lem:b1}
Let $G$ be a subcubic graph without separating cliques.
If $G$ has an induced cycle $C$ of length $p$ and an induced copy of $S_{2,2,2}$,  
then $G$ has an induced cycle of length at least $p/12$ containing the center of an induced $S_{2,2,2}$. 
\end{lemma}  

\begin{proof}
Denote a copy of an induced $S_{2,2,2}$ by $H$.
We denote the vertices of $H$ by $a_0$ (the center),
$a_1,b_1,c_1$ (vertices of degree 2), $a_2,b_2,c_2$ (vertices of degree 1 adjacent to $a_1,b_1,c_1$, respectively). 
If vertex $a_0$ belongs to $C$, there is nothing to prove. We split the rest of the proof into cases depending on the distance from vertex $a_0$ to $C$. 

\medskip

{\it Case} 1. Assume first that vertex $a_0$ is of distance $1$ from $C$. We
may suppose that $a_1\in C$. Then $a_2$ also belongs to $C$ due to the degree
constraint. If $b_1$ or $c_1$ belong to $C$ then it is easy to find an induced
cycle of length at least $p/3$ containing $a_0$. So, we assume that $b_1,c_1$
do not belong to $C$.  If $b_1$ or $c_1$ has no neighbours on $C$, then $a_1$
is the center of an induced $S_{2,2,2}$ belonging to $C$ and the result holds.
We therefore assume that both $b_1,c_1$ have neighbours in $C$.  At least one
of these neighbours is not adjacent to $a_1$ (since one of the neighbours of
$a_1$ in $C$ is $a_2$), say without loss of generality $x\in C$ is a neighbour
of $b_1$.  Then $x$ is not connected to $a_0$, due to the degree constraint. We
therefore form a long induced cycle by using $a_1a_0b_1x$ and the longer of the
two paths connects $x$ and $a_1$ in $C$.

\medskip {\it Case} 2.  Assume now that $a_0$ is of distance $2$ from $C$ and
that a shortest path from $a_0$ to $C$ goes through $a_1\not\in C$. Let $x$ be
the neighbour of $a_1$ on $C$. If $a_1$ has no other neighbour on $C$ except
for $x$, then $x$ is the center of an induced $S_{2,2,2}$ (together with
$a_0,a_1$). If $a_1$ has two non-consecutive neighbours on $C$, then $G$ has an
induced cycle of length at least $p/2$ containing $a_1$, and hence, according
to Case 1, $G$ has an induced cycle of length at least $p/6$ containing $a_0$.
If $a_1$ has two consecutive neighbours on $C$, then these neighbours together
with $a_1$ create a clique. Therefore, there must exist a path connecting $a_0$
to $C$ and avoiding this clique. But then $G$ has an induced cycle of length at
least $p/2$ containing $a_0$. 

\medskip
{\it Case} 3. Assume that $a_0$ is of distance more than 2 from $C$. To prove the result in this case, we use the notion of a quasi-chord defined as follows. 
A {\it quasi-chord for $C$} is a chordless path $P=(p_1,\ldots,p_s)$ such that each of $p_1$ and $p_s$ has two consecutive neighbours on $C$, while the other vertices of $P$ 
have no neighbours on $C$. Note that
a quasi-chord $P$ splits $C$ into two parts one of which together 
with $P$ creates an induced cycle of length at least $p/2$.

\begin{itemize}
\item[(1)] First, let us show that if $a_0$ is of distance more than 2 from $C$, then there 
exists at least one quasi-chord for $C$ with the property that the distance between $a_0$ and this quasi-chord 
is strictly less than the distance between $a_0$ and $C$. 
Since $G$ is connected, there must exist a path connecting $a_0$ to $C$. Let $P'=(x_1,\ldots,x_p)$
be a shortest path between $a_0$ and $C$ with $x_1=a_0$ and with $x_p$ 
having a neighbour in $C$. If $x_p$ has a unique neighbour on $C$, then this neighbour is the center of 
an induced $S_{2,2,2}$, in which case we are done.  If $x_p$ has two
non-consecutive neighbours on $C$, then $x_p$ is the center of an induced
$S_{2,2,2}$, in which case we are in conditions of Case 1.

Therefore, $x_p$ has two consecutive neighbours on $C$, say $c_1$ and $c_2$.
Then $x_p,c_1,c_2$ is a clique and hence there must exist a path connecting
$a_0$ to $C$ and avoiding this clique.  Let $P''=(y_1,\ldots,y_t)$ be a
shortest path of this type with $y_1=a_0$ and with $y_t$ having a neighbour in
$C$.  Then by analogy with $P'$ we conclude that $y_t$ has two consecutive
neighbours on $C$. We observe that the two paths $P'$ and $P''$ may have common
vertices different from $a_0$. Also, there may exist chords (edges) between
vertices of these paths.  However, we can always find a chordless path $P$
connecting $x_p$ to $y_t$, which uses only the vertices of $P'$ and $P''$ by
considering a shortest path in $G[V(P') \cup V(P'')]$, the graph induced by the
vertices of the two paths.  This path $P$ is a quasi-chord for $C$.  Moreover,
$P$ is closer to $a_0$ than $C$, since $P$ contains $x_p$ and by definition the
distance between $a_0$ and $x_p$ is exactly one less than the distance between
$a_0$ and $C$.  
 
\item[(2)] Now let us show that there exists a quasi-chord for $C$ which is of distance at most $2$ from $a_0$. To this end, let us denote by $P=(p_1,\ldots,p_s)$ 
a quasi-chord for $C$ which is as close to $a_0$ as possible. 

Assume $a_0$ is of distance more than $2$ from $P$. Then we consider a shortest path $P'=(x_1,\ldots,x_p)$ connecting $a_0=x_1$ to $P$ with $x_p$ having a neighbour in $P$.
Since $P$ is closer to $a_0$ than $C$, no vertex of $P'$ belongs 
to $C$ or has a neighbour on $C$. Also, since $P'$ is shortest, no vertex of $P'$ has a 
neighbour on $P$, except for $x_p$. We recall that the quasi-chord $P$ splits $C$ into two parts one of which
together with $P$ creates an induced cycle $C'$ of length at least $p/2$. 
To avoid an induced $S_{2,2,2}$ whose center is of distance at most 1 from $C'$ 
(in which case we are in conditions of Case 1 with respect to the induced cycle $C'$), 
we conclude that $x_p$ has two consecutive neighbours on $P$, say $p_i$ and $p_{i+1}$. 
Using the fact $G$ has no separating clique, we find one more chordless path $P''=(y_1,\ldots,y_t)$
connecting $a_0=y_1$ to $C\cup P$ and avoiding the clique $\{x_p,p_i,p_{i+1}\}$. As before, we assume that no vertex of $P''$ except for $y_t$ has a neighbour on $C\cup P$.

Let $P^*$ be a chordless path connecting $x_p$ to $y_t$ and consisting of vertices of $P'$ and $P''$ only. If $y_t$ has no neighbours on $C$, then a part of $P$ can be replaced 
by $P^*$ creating a quasi-chord for $C$, which is closer to $a_0$ than $P$, since it contain $x_p$. This contradicts the choice of $P$. Therefore, $y_t$ has a neighbour on $C$.
To avoid an induced $S_{2,2,2}$ whose center is of distance at most 1 from $C$, we conclude that $y_t$ has two consecutive neighbours on $C$. But then $\{p_1,\ldots,p_i\}\cup P^*$
is  a quasi-chord for $C$, which is closer to $a_0$ than $P$, since it contain $x_p$. This final contradiction shows that the distance from $a_0$ to $P$ is at most 2.
\end{itemize}

From Claim (2) we conclude that $G$ has an induced cycle of length at least $p/2$, which is of distance at most 2 from $a_0$. Therefore, according to Cases 1 and 2, 
$G$ has an induced cycle of length at least $p/12$ containing the center of an induced $S_{2,2,2}$.
\end{proof}

\begin{lemma}\label{lem:b2}
Let $G$ be a subcubic graph without separating cliques. 
If $G$ has an induced cycle $C$ of length $p$ containing the center of an induced $S_{2,2,2}$, 
then $G$ has an induced extended cycle $C^*_t$ or an induced apple with a long stem $A^*_t$ with $t\ge p/8$. 
\end{lemma}  

\begin{proof}
Denote the induced $S_{2,2,2}$ with the center on $C$ by $H$.
We also denote the vertices of $H$ by $0$ (the center),
$1,2,3$ (vertices of degree 2), $4,5,6$ (vertices of degree 1 adjacent to 1,2,3, respectively) and assume 
without loss of generality that $1$ and $2$ belong to the cycle (together with 0), while $3$ does not belong to $C$. 
Finally, we denote the two vertices of $C$ following vertex $1$ by $a$ and $b$
(possibly $a=4$), and the two vertices of $C$ following vertex $2$ by $c$ and $d$ (possibly $c=5$).
Let us call the number of edges of $H$ contained in $C$ the $H$-value of $C$. Clearly, this value cannot be larger than 4
and due to degree constraint it cannot be smaller than 2. 

The following two claims will be useful in the proof of the lemma. The proof of the first claim is evident.
\begin{claim}\label{claim:A}
If a vertex $x \notin C$ 
\begin{itemize}
\item has two neighbours on $C$, then the smaller of the two cycles formed by $x$ and the two parts of $C$ has size at most 5,
\item has three neighbours on  $C$, then  the two smaller cycles (of the three formed by $x$ and the three parts of $C$) 
have size at most 4,
\end{itemize}
since otherwise an induced $A^*_{t}$ with $t\ge p/3$ can be easily found.
\end{claim}

\begin{claim}\label{claim:B}
Either vertex $6$ belongs to $C$ or $6$ has a neighbour on $C$ or $6$ is adjacent to a vertex $x$ that has a neighbour on $C$.
\end{claim}

\begin{proof}
Assume vertex $6$ neither belongs to $C$ nor has a neighbour on $C$. Then vertex $3$ must have a neighbour on $C$,
for otherwise $C \cup \{3,6\}$ is an induced $A^*_{p}$. By Claim~\ref{claim:A} (and the fact that $6\not\in C$), 
vertex $3$ cannot have neighbours outside of $\{a,b,c,d\}$. Also, since $6\not\in C$,
$3$ cannot have neighbours both in $\{a,b\}$ and in $\{c,d\}$. 
Therefore, we may assume without loss of generality that $3$ has a neighbour in $\{a,b\}$ and
does not have neighbours in $\{c,d\}$.

Since $G$ has no separating cliques, vertex $6$ must have a neighbour $x$ different from $3$. Then $x$ has a neighbour on $C$, 
for otherwise $C' \cup \{6,x\}$ is an induced $A^*_{p-1}$, where $C'$ is the cycle formed by $3$ and the long part of $C$.    
\end{proof}

\medskip
First, we prove the lemma assuming that the $H$-value of $C$ is 4, i.e. $4=a$ and $5=c$, 
and show that in this case $G$ has an induced $C^*_t$ or an induced $A^*_t$ with $t\ge p/2$. 
\begin{itemize}
\item[(1.1)] If vertex $3$ is adjacent neither to $b$ nor to $d$, then $3$ has no other neighbours on $C$ (except for $0$) by Claim~\ref{claim:A}.
In particular, $6\not\in C$. But then $6$ has a neighbour on $C$, since otherwise $C_{3,6}$ is an induced $A^*_{p}$.
If $6$ has two neighbours on $C$ or if $6$ has a single neighbour different from $b$ and $d$, then $G$ contains an induced 
$A^*_t$ with $t\ge p/2$. If $b$ (or $d$) is the only neighbour of $6$ on $C$, 
then $C\cup \{3,6\}$ is an induced $C^*_p$.

\item[(1.2)] Now assume without loss of generality that $3$ is adjacent to $b$. 
Then $3$ has no other neighbours on $C$ (except for $0$ and $b$) by Claim~\ref{claim:A}.
In particular, $6\not\in C$. If $6$ has neighbors on $C$, we denote them by $y$ and $z$ (possibly $y=z$).
Otherwise, by Claim \ref{claim:B}, $6$ is adjacent to a vertex $x$ which has neighbours on $C$.
In this case, we denote the neighbours of $x$ on $C$ by $y$ and $z$ (again, possibly $y=z$).
The neighbours of $3$ on $C$ split the cycle into two paths: the short one of length $4$ and a longer one
of length $p-2$. Let us denote the longer path by $P$.
\begin{itemize}
\item[(1.2.1)] Suppose first that neither $y$ nor $z$ belongs to $\{1,4\}$. Then
$y$ and $z$ split $P$ into at most 3 paths giving rise to at most 3 cycles.
It is not difficult to see that each of these cycles is part of an apple with a long stem,
and one of these cycles has length at least $p/2$ (remember that the distance between $y$ and $z$ along $C$ cannot be larger
than $3$ by Claim~\ref{claim:A}), i.e. $G$ contains an induced $A^*_t$ with $t\ge p/2$.

\item[(1.2.2)] Now assume without loss of generality that $y\in\{1,4\}$. Then,
$y$ is a neighbour of $x$, since $6$ has no neighbours among $\{1,4\}$.
Therefore, $6$ has no neighbours in $C$. We first observe that we may assume
without loss of generality that $y=1$. This is because the vertices
$\{1,3,4,6,b\}$ together with the two other vertices at distance at most two
from $b$ in $C$ induce an $S_{2,2,2}$ with center $b$ and $H$-value 4. In case
$y=4$,  by exchanging the roles of $b$ with $0$, $1$ with $4$, and $\{2,5\}$
with the two other vertices at distance at most two from $b$ in $C$ we have
that $x$ (the neighbour of $6$ with connection to $C$) is adjacent to $1$.

Furthermore, we observe that $z$ (the second neighbour of $x$ in $C$) must
exist and belong in $C\setminus\{0,1,4,b\}$: otherwise
$(C\setminus\{1,4\})\cup\{3,6,x\}$ induce an $A^*_{p-2}$.

To summarize, we have that the neighbours of $3$ in $C$ are $\{0,b\}$;
$6\not\in C$ and $6$ has no neighbours in $C$; and $6$ has a neighbour $x\neq3$
which is connected to $1$ and to a vertex $z\in C\setminus\{0,1,4,b\}$. Now, by
Claim~\ref{claim:A}, $z$ can be at distance at most $3$ from $1$ in $C$.
Therefore, either $z\in\{2,5\}$ or $z$ is the neighbor of $b$ in $C$ which is
not $4$. In the latter case we have that $(C\setminus\{4\})\cup\{3,x\}$ induce
a $C^*_{p-1}$, formed by the long induced path from $0$ to $z$ on $C$, plus the
two paths $z,b,3,0$ and $z,x,1,0$. In the former case we see that an induced
$C^*_t$ for $t\ge p-1$ is formed by the long path from $b$ to $z$ on $C$, the
vertex $x$ (which is a neighbour of $z$) and the two paths $b,3,6,x$ and
$b,4,1,x$.

\end{itemize} \end{itemize}

To complete the proof of the lemma, we prove two more claims that eliminate the cases when the $H$-value of $C$ is 3 or 2.

\begin{claim}\label{claim:D}
If the $H$-value of $C$ is $3$ with $4\ne a$, then $G$ contains either a cycle of length at least $p/2$ with $H$-value $4$ or an induced $A^*_t$ with $t\ge p/2$ or an induced $C^*_p$.
\end{claim}
\begin{proof}
We prove the result through a series of claims.
\begin{itemize}

\item {\it $a$ is adjacent to $3$}. Indeed, if $a$ is adjacent neither to $3$
nor to $6$, then by replacing $4$ with $a$ we obtain an induced $S_{2,2,2}$
containing 4 edges in $C$.  Suppose then that $a$ is adjacent to $6$ but not
$3$. We distinguish two cases: first, $b=6$, in which case the cycle induced by
$C\cup\{3\}\setminus\{1,a\}$ has length $p-1$, contains $0$ and has $H$-value
4; and second $6\not\in C$.  In this case, any other neighbour of $6$ on $C$
(if any) must be of distance at most 3 from $a$ (Claim~\ref{claim:A}),
therefore we obtain a cycle of length at least $p-2$ containing four edges
$(0,3),(3,6),(0,2),(2,5)$ of $H$.

\item {\it $6$ does not belong to $C$}. Indeed, if $6$ belongs to $C$, then it must be of distance at most 3 from $a$  (Claim~\ref{claim:A}), 
in which case we obtain a cycle of length at least $p-3$ containing four edges $(0,3),(3,6),(0,2),(2,5)$ of $H$.
\item {\it $4$ has no neighbours on $C$ different from $1$}. Indeed, if $4$ has more neighbours on $C$, then by Claim~\ref{claim:A} the farthest neighbour must be of distance 
at most 3 (if $|N(4)\cap C|=2$) or at most 4 (if $|N(4)\cap C|=3$) from $1$,  in which case we obtain a cycle of length at least $p-2$ containing four edges $(0,1),(1,4),(0,2),(2,5)$ of $H$.
\item {\it $4$ is  adjacent to a vertex $u\not\in C$ that has a neighbour on $C$}, since otherwise either $1$ is a separating clique or $C_{4,u}=A^*_{p}$.
\item {\it $u$ is adjacent to $2$}. Indeed, assume $u$ is not adjacent to $2$. By Claim~\ref{claim:A} the neighbours of $u$ on $C$ must be of distance at most 3 from each other. Therefore, $4,u$ together with a large part of $C$
create a cycle $C'$ of length at least $p/2$ containing either the edge $(0,1)$ or the edge $(1,a)$. If $(0,1)\in C'$, then the $H$-value of $C'$ is 4 (it contains the four edges $(0,1),(1,4),(0,2),(2,5)$).
If $(1,a)\in C'$, then either $C\cup \{4,u\}=C^*_p$ (if $5$ is the only neighbour of $u$ on $C$) or $C'_{0,2}=A^*_{t}$ with $t\ge p/2$ (if $u$ is adjacent to a vertex of $C$ different from $5$).

\item {\it no vertex of $N(4)\setminus\{u,1\}$ has a neighbour in $C$}, since
if there exists a $v\not\in\{u,1\}$ that is a neighbour of $4$ and has a
neighbour in $C$, then we have $v\not\in C$ and similarly to the previous claim
we would conclude that $v$ is adjacent to $2$. However, this is impossible due
to the degree constraint.

\end{itemize}
From Claim~\ref{claim:B} and the above discussion we conclude that either $6$ has neighbours on $C$ or  $6$ is adjacent to a vertex $x$ that has neighbours on $C$. These neighbours (of $6$ or of $x$) 
must be located on $C$ at distance at most 3 from each other (Claim~\ref{claim:A}) and must be different from $0,1,a$. 
Therefore, vertices $3,6$ (and possibly $x$) together with a long part of $C$ create a cycle $C'$ of length at least $p/2$
containing either the edge $(0,3)$ or the edge $(3,a)$.  In both cases,
$C'_{1,4}=A^*_{t}$ with $t\ge p/2$, unless $x$ is adjacent to $4$. But in the
latter case $x=u$ (where $u$ is the neighbour of $4$ adjacent to $2$) and hence
$C_{x,6}=A^*_p$.  \end{proof}

\begin{claim}\label{claim:C}
If the $H$-value of $C$ is $2$, then $G$ contains either a cycle of length at least $p/2$ with $H$-value $3$ or an induced $A^*_t$ with $t\ge p/2$ or an induced $C^*_{p-1}$.
\end{claim}
\begin{proof}

Let us first establish that both $a$ and $c$ must have a neighbor in $\{3,6\}$.
Suppose that $c$ does not have a neighbor in $\{3,6\}$. We now observe that $a$
must have a neighbor in $\{3,6\}$, for otherwise we exchange $\{4,5\}$ with
$\{a,c\}$ and obtain $H$-value $4$. Then, if $c$ is not adjacent to $4$, we can
exchange $c$ with $5$ in $H$ to increase the $H$-value of $C$.  Suppose then
that $c$ is adjacent to $4$. By Claim~\ref{claim:A} $c$ and $1$ are the only
neighbours of $4$ on $C$.  If $5$ has no neighbor in $C$ besides $2$, then we
have an induced  $A^*_{p-1}=C\cup\{4,5\}\setminus\{0\}$.  If $5$ is adjacent to
$a$, then by Claim~\ref{claim:A} $a$ is the only neighbour of $5$ on $C$, in
which case we obtain an induced $C^*_{p-1}=C\cup\{4,5\}\setminus\{0\}$.  If $5$
is not adjacent to $a$, then any neighbour of $5$ on $C$ must be of distance at
most $3$ from $c$. We therefore have an induced cycle of length at least $p-2$
that goes through $0,2,5$, then continues to the neighbor of $5$ on $C$ that is
as far as possible from $c$, and then goes on to $0$ by using vertices of $C$
and avoiding $c$. This induced cycle has a higher $H$-value, as it contains
edges $(0,1), (0,2)$ and $(2,5)$. 

Because of the above, both $a$ and $c$ must have a neighbor in $\{3,6\}$. They
cannot both be adjacent to $6$ by Claim~\ref{claim:A}, and they cannot both be
adjacent to $3$, for otherwise, without loss of generality, $6=a$ and hence $6$
is adjacent to $1$, which is impossible.  Therefore, we assume, without loss of
generality, that $a$ is adjacent to $6$ and $3$ is adjacent to $c$.

Now we look at vertex $4$. If $4$ has neighbours on $C$ different from $1$,
then these neighbours must be close to $1$ (by Claim~\ref{claim:A}), in which
case we find a cycle of length at least $p-2$ containing both edges $(0,1)$ and
$(1,4)$ and hence having $H$-value at least $3$. Note that we are also using
here the fact that $4$ is not connected to $a$ due to the degree constraint,
since $a$ is connected to $6$ and two vertices of $C$.

If $4$ has no other neighbours on $C$, then it must be adjacent to a vertex $u$
different from $1$ (since otherwise vertex $1$ is a separating clique) and $u$
must have a neighbour on $C$ (to avoid a large apple with a long stem
$C_{4,u}$). By Claim~\ref{claim:A}, the neighbours of $u$ on $C$ must be close
to each other, and again may not include $a,c$ due to the degree constraint and
the fact that $u\not\in\{3,6\}$, as $u$ is adjacent to $4$.  Therefore,
vertices $4$ and $u$ together with two parts of $C$ form two cycles, one of
which is large, i.e.  has length at least $p/2$.  If the large cycle contains
both edges $(0,1)$ and $(1,4)$, then it has $H$-value at least $3$. If the
large cycle does not contain the edge $(0,1)$, then this cycle together with
the vertices $0,2$ forms a large apple with a long stem $A^*_t$ with $t\ge
p/2$.  \end{proof}

Summarizing, we conclude that if $G$ has an induced cycle $C$ of length $p$ containing the center of an induced $S_{2,2,2}$, then $G$ contains either an induced $C^*_t$ or an induced $A^*_t$ with $t\ge p/8$. 
\end{proof}


\section{Destroying large extended cycles}
\label{sec:destroy}


According to the previous section, if an $(A^*_k,A^*_{k+1},\ldots)$-free subcubic graph $G$ contains a large induced cycle and an induced copy of $S_{2,2,2}$, then it must contain a large extended cycle $C^*$.
The goal of the present section is to show how to destroy large extended cycles by means of various local graph reductions. 
We describe these reductions in Section~\ref{sec:reductions} and apply them to large extended cycles in Section~\ref{sec:analysis}.

\subsection{Graph reductions}
\label{sec:reductions}

\subsubsection{$\Phi$-reduction and $house$-reduction}

We start with the $\Phi$-reduction introduced in \cite{subcubic}. It applies to a graph $G$ containing an induced 
copy of the graph $\Phi$ represented on the left of Figure~\ref{fig:Phi} and consists in replacing $\Phi$ by the graph on the right of Figure~\ref{fig:Phi}.

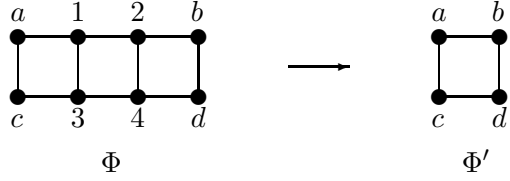
\begin{figure}[ht]
\begin{center}
\setlength{\unitlength}{0.2cm}
\begin{picture}(40,9)

\put(4,2){\circle*{1}}
\put(8,2){\circle*{1}}
\put(12,2){\circle*{1}}
\put(16,2){\circle*{1}}

\put(4,6){\circle*{1}}
\put(8,6){\circle*{1}}
\put(12,6){\circle*{1}}
\put(16,6){\circle*{1}}

\put(32,2){\circle*{1}}
\put(36,2){\circle*{1}}
\put(32,6){\circle*{1}}
\put(36,6){\circle*{1}}

\put(4,6){\line(1,0){12}}
\put(4,2){\line(1,0){12}}

\put(8,6){\line(0,-1){4}}
\put(4,6){\line(0,-1){4}}
\put(12,6){\line(0,-1){4}}
\put(16,6){\line(0,-1){4}}

\put(32,6){\line(1,0){4}}

\put(32,2){\line(1,0){4}}
\put(32,6){\line(0,-1){4}}

\put(36,6){\line(0,-1){4}}

\put(22,4){\vector(1,0){4}}

\put(3.5,0){$c$}
\put(3.5,7){$a$}

\put(7.5,0){$3$}
\put(7.5,7){$1$}

\put(11.5,0){$4$}
\put(11.5,7){$2$}

\put(15.5,0){$d$}
\put(15.5,7){$b$}

\put(31.5,0){$c$}
\put(31.5,7){$a$}

\put(35.5,0){$d$}
\put(35.5,7){$b$}

\put(9.5,-3){$\Phi$}
\put(33.5,-3){$\Phi'$}

\end{picture}
\end{center}
\caption{$\Phi$-reduction}
\label{fig:Phi}
\end{figure}

\begin{lemma}\label{lem:phi}
By applying the $\Phi$-reduction to an $(A^*_k,A^*_{k+1},\ldots)$-free subcubic graph $G$, we obtain an $(A^*_k,A^*_{k+1},\ldots)$-free subcubic graph $G'$ with 
$\alpha(G')=\alpha(G)-2$. 
\end{lemma}

\begin{proof}
The equality $\alpha(G')=\alpha(G)-2$ was proved in \cite{subcubic}. To prove $(A^*_k,A^*_{k+1},\ldots)$-freeness of $G'$ assume by contradiction 
that $G'$ contains an induced copy $H$ of a large apple with a long stem $A^*_p$ with $p\ge k$. Then  at least one of the four vertices $a,b,c,d$, say $d$, does not belong to $H$,
since otherwise $H$ contains a $C_4$. 
But then the vertices inducing $H$ together with vertices $1$ and $2$ induce a subdivision of $A^*_p$ in $G$, which is impossible. 
\end{proof}

\medskip
A $house$ is the complement of a $P_5$. If a graph $G$ contains an induced $house$, the $house$-reduction consists in removing from $G$ the vertices 
that form a triangle in the $house$. It was shown in \cite{subcubic} that if $G$ is a subcubic graph, then the $house$-reduction reduces $\alpha(G)$ 
by exactly~1.

\subsubsection{$\Pi$-reduction}

Now we introduce the $\Pi$-reduction illustrated in Figure~\ref{fig:pi}. 
In a graph $G$, an induced $\Pi$ is the graph represented on the left of
Figure~\ref{fig:pi}.  We observe that vertex $f$ can be missing, in which case
vertices $a$ and $c$ have no other neighbours in $G$. However, if $f$ exists,
that is, if one of $a,c$ has a neighbour outside of $\{1,3,e\}$, then $f$ is a
common neighbour of $a,c$. Similarly, vertex $h$ can be missing, in which case
vertices $b$ and $d$ have no other neighbours in $G$.

\begin{figure}[ht]
\begin{center}
\setlength{\unitlength}{0.2cm}
\begin{picture}(40,9)
\put(0,2){\circle*{1}}
\put(4,2){\circle*{1}}
\put(8,2){\circle*{1}}
\put(12,2){\circle*{1}}
\put(16,2){\circle*{1}}
\put(20,2){\circle*{1}}
\put(0,6){\circle*{1}}
\put(4,6){\circle*{1}}
\put(8,6){\circle*{1}}
\put(12,6){\circle*{1}}
\put(16,6){\circle*{1}}
\put(20,6){\circle*{1}}
\put(28,2){\circle*{1}}
\put(32,2){\circle*{1}}
\put(36,2){\circle*{1}}
\put(40,2){\circle*{1}}
\put(28,6){\circle*{1}}
\put(32,6){\circle*{1}}
\put(36,6){\circle*{1}}
\put(40,6){\circle*{1}}

\put(4,2){\line(1,0){12}}
\put(0,6){\line(1,0){20}}
\put(32,2){\line(1,0){4}}
\put(28,6){\line(1,0){12}}

\put(0,6){\line(1,-1){4}}
\put(8,6){\line(1,-1){4}}

\put(28,6){\line(1,-1){4}}
\put(40,6){\line(-1,-1){4}}

\put(12,6){\line(-1,-1){4}}
\put(20,6){\line(-1,-1){4}}

\multiput(0,2)(0.5,0){8}{\circle*{0.1}}
\multiput(0,2)(0.5,0.5){8}{\circle*{0.1}}

\multiput(20,2)(-0.5,0){8}{\circle*{0.1}}
\multiput(20,2)(-0.5,0.5){8}{\circle*{0.1}}

\multiput(28,2)(0.5,0){8}{\circle*{0.1}}
\multiput(28,2)(0.5,0.5){8}{\circle*{0.1}}

\multiput(40,2)(-0.5,0){8}{\circle*{0.1}}
\multiput(40,2)(-0.5,0.5){8}{\circle*{0.1}}
\put(22,4){\vector(1,0){4}}
\put(-0.5,0){$f$}
\put(-0.5,7){$e$}

\put(3.5,0){$c$}
\put(3.5,7){$a$}

\put(7.5,0){$3$}
\put(7.5,7){$1$}

\put(11.5,0){$4$}
\put(11.5,7){$2$}

\put(15.5,0){$d$}
\put(15.5,7){$b$}

\put(19.5,0){$h$}
\put(19.5,7){$g$}

\put(27.5,0){$f$}
\put(27.5,7){$e$}

\put(31.5,0){$c$}
\put(31.5,7){$a$}

\put(35.5,0){$d$}
\put(35.5,7){$b$}

\put(39.5,0){$h$}
\put(39.5,7){$g$}

\put(9.5,-3){$\Pi$}
\put(33.5,-3){$\Pi'$}
\end{picture}
\end{center}
\caption{$\Pi$-reduction}
\label{fig:pi}
\end{figure}
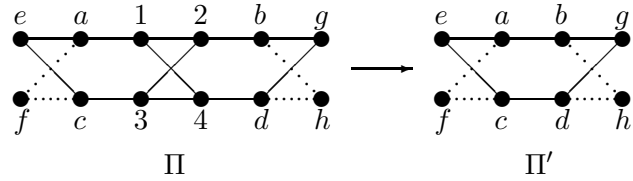

\begin{lemma}\label{lem:pi}
By applying the $\Pi$-reduction to an $(A^*_k,A^*_{k+1},\ldots)$-free subcubic graph $G$, we obtain an $(A^*_k,A^*_{k+1},\ldots)$-free subcubic graph $G'$ with 
$\alpha(G')=\alpha(G)-2$. 
\end{lemma}

\begin{proof}
Let $S$ be a maximum independent set in $G$ and $X=S\cap \{1,2,3,4\}$, $Y=S\cap \{a,b,c,d\}$.
If $|Y|=4$, then $|X|=0$ and hence $S-\{a,c\}$ is an independent set in $G'$ of size $\alpha(G)-2$.
If $|Y|=3$, say $Y=\{a,b,c\}$, then $X=\{4\}$ and hence $S-\{4,b\}$ is an independent set in $G'$ of size $\alpha(G)-2$.

Let $|Y|=2$. Then, up to symmetry, $Y=\{a,b\}$ or $Y=\{a,d\}$ or $Y=\{a,c\}$. If $Y=\{a,b\}$ or $Y=\{a,d\}$
then $|X|=1$ and hence $S-(X\cup\{a\})$ is an independent set in $G'$ of size $\alpha(G)-2$.
If $Y=\{a,c\}$ or $|Y|\le 1$, then $|X|=2$ and hence $S-X$ is an independent set in $G'$ of size $\alpha(G)-2$.
Therefore, $\alpha(G')\ge\alpha(G)-2$.

Conversely, let $S$ be a maximum independent set in $G'$ and $Y=S\cap \{a,b,c,d\}$. Clearly, $|Y|\le 2$ and if $|Y|=2$ we can assume
without loss of generality that $Y=\{a,d\}$ or $Y=\{a,c\}$. 

If $|Y|\le 1$, then $S$ can be always extended by adding two vertices from $\{1,2,3,4\}$ to an independent set of size $\alpha(G')+2$ in $G$.
If  $Y=\{a,d\}$, then $g,h\not\in S$ and hence $S\cup \{b,3\}$ is an independent set of size $\alpha(G')+2$ in $G$.
If  $Y=\{a,c\}$, then $S\cup \{2,4\}$ is an independent set of size $\alpha(G')+2$ in $G$.
Therefore, $\alpha(G')+2\le\alpha(G)$ and hence $\alpha(G')=\alpha(G)-2$.

To prove $(A^*_k,A^*_{k+1},\ldots)$-freeness of $G'$ assume by contradiction 
that $G'$ contains an induced copy $H$ of a large apple with a long stem $A^*_p$ with $p\ge k$.
Then $H$ contains at least one of the edges $ab$ and $cd$, since otherwise the same vertices induced $H$ in $G$.

If $H$ contains both edges $ab$ and $cd$, then exactly one of the vertices $e,f,g,h$ belongs to $H$.
Indeed, if two of these vertices belong to $H$, then $H$ contains a small cycle, and  if none of them belongs to $H$, then $H$ is not connected.
Assuming, without loss of generality, that $e$ belongs to $H$, we conclude that $H$ contains two vertices $b$ and $d$, each of which has degree 1 in $H$, which is impossible.

Now assume the edge $ab$ belongs to $H$ and the edge $cd$ does not. Without loss of generality, $d\not\in V(H)$.
If $c\in V(H)$, then $e\not\in V(H)$ or $f\not\in V(H)$, 
since otherwise $H$ contains a $C_4$, say $f\not\in V(H)$. Then $c$ has degree 1 in $H$ and hence vertex $a$ must have degree 3 in $H$. But $a$ has degree 2 in $H$, a contradiction.

If $c\not\in V(H)$, then the vertices inducing $H$ together with vertices $1$ and $2$ induce a subdivision of $A^*_p$ in $G$, which is impossible. 
\end{proof}

\subsubsection{$\Gamma$-reduction}

One more reduction is illustrated in Figure~\ref{fig:gamma}. We will refer to
it as $\Gamma$-reduction.  Again, vertex $f$ can be missing,  in which case
vertices $b$ and $d$ have degree 2 in the graph, but if $f$ exists it is a
common neighbour of $b,d$.

\begin{figure}[ht]
\begin{center}
\setlength{\unitlength}{0.2cm}
\begin{picture}(40,9)
\put(4,2){\circle*{1}}
\put(8,2){\circle*{1}}
\put(12,2){\circle*{1}}
\put(16,2){\circle*{1}}
\put(20,2){\circle*{1}}
\put(4,6){\circle*{1}}
\put(8,6){\circle*{1}}
\put(12,6){\circle*{1}}
\put(16,6){\circle*{1}}
\put(20,6){\circle*{1}}
\put(32,2){\circle*{1}}
\put(36,2){\circle*{1}}
\put(40,2){\circle*{1}}
\put(32,6){\circle*{1}}
\put(36,6){\circle*{1}}
\put(40,6){\circle*{1}}

\put(4,2){\line(1,0){12}}
\put(4,6){\line(1,0){16}}
\put(32,2){\line(1,0){4}}
\put(32,6){\line(1,0){8}}

\put(8,6){\line(1,-1){4}}

\put(40,6){\line(-1,-1){4}}

\put(12,6){\line(-1,-1){4}}
\put(20,6){\line(-1,-1){4}}

\put(4,6){\line(0,-1){4}}
\put(32,6){\line(0,-1){4}}

\multiput(20,2)(-0.5,0){8}{\circle*{0.1}}
\multiput(20,2)(-0.5,0.5){8}{\circle*{0.1}}


\multiput(40,2)(-0.5,0){8}{\circle*{0.1}}
\multiput(40,2)(-0.5,0.5){8}{\circle*{0.1}}
\put(24,4){\vector(1,0){4}}

\put(3.5,0){$c$}
\put(3.5,7){$a$}

\put(7.5,0){$3$}
\put(7.5,7){$1$}

\put(11.5,0){$4$}
\put(11.5,7){$2$}

\put(15.5,0){$d$}
\put(15.5,7){$b$}

\put(19.5,0){$f$}
\put(19.5,7){$e$}


\put(31.5,0){$c$}
\put(31.5,7){$a$}

\put(35.5,0){$d$}
\put(35.5,7){$b$}

\put(39.5,0){$f$}
\put(39.5,7){$e$}

\put(11.5,-3){$\Gamma$}
\put(35.5,-3){$\Gamma'$}
\end{picture}
\end{center}
\caption{$\Gamma$-reduction}
\label{fig:gamma}
\end{figure}
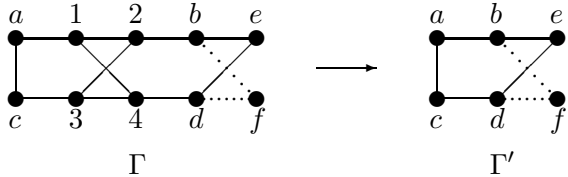

\begin{lemma}\label{lem:gamma}
By applying the $\Gamma$-reduction to an  $(A^*_k,A^*_{k+1},\ldots)$-free subcubic graph $G$, we obtain an $(A^*_k,A^*_{k+1},\ldots)$-free subcubic graph $G'$ with 
$\alpha(G')=\alpha(G)-2$. 
\end{lemma}

\begin{proof}
Let $S$ be a maximum independent set in $G$ and $X=S\cap \{1,2,3,4\}$, $Y=S\cap \{a,b,c,d\}$.
If $|Y|=3$, say $Y=\{a,b,d\}$, then $X=\{3\}$ and hence $S-\{a,3\}$ is an independent set in $G'$ of size $\alpha(G)-2$.

Let $|Y|=2$. Then, up to symmetry, $Y=\{a,b\}$ or $Y=\{a,d\}$ or $Y=\{b,d\}$. If $Y=\{a,b\}$ or $Y=\{a,d\}$
then $|X|=1$ and hence $S-(X\cup\{a\})$ is an independent set in $G'$ of size $\alpha(G)-2$.
If $Y=\{b,d\}$ or $|Y|\le 1$, then $|X|=2$ and hence $S-X$ is an independent set in $G'$ of size $\alpha(G)-2$.
Therefore, $\alpha(G')\ge\alpha(G)-2$.

Conversely, let $S$ be a maximum independent set in $G'$ and $Y=S\cap \{a,b,c,d\}$. Clearly, $|Y|\le 2$ and if $|Y|=2$ we can assume
without loss of generality that $Y=\{a,d\}$ or $Y=\{b,d\}$. 

If $|Y|\le 1$ or $Y=\{b,d\}$, then $S$ can be extended by adding two vertices from $\{1,2,3,4\}$ to an independent set of size $\alpha(G')+2$ in $G$.
If  $Y=\{a,d\}$, then $e,f\not\in S$ and hence $S\cup \{b,3\}$ is an independent set of size $\alpha(G')+2$ in $G$.
Therefore, $\alpha(G')+2\le\alpha(G)$ and hence $\alpha(G')=\alpha(G)-2$.

To prove $(A^*_k,A^*_{k+1},\ldots)$-freeness of $G'$ assume by contradiction 
that $G'$ contains an induced copy $H$ of a large apple with a long stem $A^*_p$ with $p\ge k$.
Then $H$ contains at least one of the edges $ab$ and $cd$, since otherwise the same vertices induced $H$ in $G$.

If $H$ contains both edges $ab$ and $cd$, then neither $e$ nor $f$ belongs to $H$,
since otherwise  $H$ contains a small cycle. Then both $b$ and $d$ have degree 1 in $H$, which is impossible.

If the edge $ab$ belongs to $H$ and the edge $cd$ does not, 
then the vertices inducing $H$ together with vertices $1$ and $2$ induce a subdivision of $A^*_p$ in $G$, which is impossible. 
\end{proof}

\subsubsection{$\Theta$-reduction}

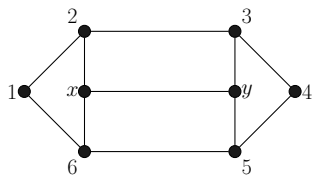
\begin{figure}[ht]
\begin{center}
\begin{tikzpicture}[scale=0.4, transform shape]
\tikzstyle{sommet}=[circle, draw,fill=black!90, inner sep=1pt, inner sep=4pt, minimum size=0.1cm]

\node[sommet] (x) at (0,0) {};
\node[sommet] (2) at (0,2) {};
\node[sommet] (6) at (0,-2) {};
\node[sommet] (1) at (-2,0) {};
\node[sommet] (y) at (5,0) {};
\node[sommet] (3) at (5,2) {};
\node[sommet] (5) at (5,-2) {};
\node[sommet] (4) at (7,0) {};

\draw (6)--(1)--(2)--(x)--(6)--(5);
\draw (3)--(4)--(5)--(y)--(3)--(2);
\draw (x)--(y);

\node () at (-0.4,0) {\huge $x$};
\node () at (-0.4,2.5) {\huge $2$};
\node () at (-0.4,-2.5) {\huge $6$};
\node () at (-2.4,0) {\huge $1$};

\node () at (5.4,0) {\huge $y$};
\node () at (5.4,2.5) {\huge $3$};
\node () at (5.4,-2.5) {\huge $5$};
\node () at (7.4,0) {\huge $4$};
\end{tikzpicture}
\end{center}
\caption{$\Theta$ graph}
\label{fig:theta}
\end{figure}

\begin{lemma}
If a subcubic graph $G$ contains an induced $\Theta$ (see Figure~\ref{fig:theta}), then the deletion of 
vertices $x,y$ reduces the independence number of $G$ by exactly 1.
\end{lemma}

\begin{proof}
Let $S$ be a maximum independent set in $G$.
Clearly, $S$ contains at most one vertex in $\{x,y\}$.
Now let us show that $S$ contains at least one vertex in $\{x,y\}$.
Assume $x,y\not\in S$. Then $x$ has a neighbour in $S$ and $y$ has a neighbours in $S$.
Without loss of generality let $2,5\in S$. But then $S$ is not maximum, since by replacing $2$ with $3,x$ 
we obtain a larger independent set.   
\end{proof}

\subsubsection{Total struction and subgraph reduction}

Total struction is an operation that was introduced in
\cite{struction-revisited}. Roughly speaking, this operation allows us to
identify a part of the graph that can be replaced by an auxiliary graph in a
way that decreases the size of the maximum independent set by a precise value.
Even though this operation is quite powerful, in this paper we will only need
to use to special cases of total struction, given by Corollaries \ref{lem:hred}
and \ref{lem:k23}. To keep the presentation self-contained we give proofs for
both transformations which also show that we do not create new forbidden
subgraphs.

\begin{corollary}\label{lem:hred}

For any graph $G=(V,E)$ and $H\subseteq V$ let $N[H]$ denote the set of
vertices at distance at most $1$ from $H$. Then, we have the following: if
$\alpha(G[H])=\alpha(G[N[H]])$, then $\alpha(G[V\setminus
N[H]])=\alpha(G)-\alpha(G[H])$.

\end{corollary}

\begin{proof}

Let us first observe that $\alpha(G)\ge \alpha(G[V\setminus N[H]]) +
\alpha(G[H])$ because that union of an independent set of $G[H]$ with an
independent set of $G[V\setminus N[H]]$ is an independent set of $G$. For the
other direction, we note that $\alpha(G)\le \alpha(G[V\setminus N[H]]) +
\alpha(G[N[H]])$. To see this, take an independent set $S$ in $G$ and observe
that $|S| = |S\setminus N[H]| + |S\cap N[H]| \le \alpha(G[V\setminus N[H]]) +
\alpha(G[N[H]])$. However, since $\alpha(G[N[H]]) = \alpha(G[H])$ we obtain the
lemma.  \end{proof}

Informally, Corollary \ref{lem:hred} gives rise to the following
transformation: if we can find a set of vertices $H$ such that $G[H]$ and
$G[N[H]]$ have the same maximum independent set, then we simply select an
independent set of $H$ in our solution and delete all vertices of $N[H]$.  The
deletion of $N[H]$ in the case when $\alpha(G[H])=\alpha(G[N[H]])$ was called
in \cite{subcubic} the $H$-subgraph reduction.

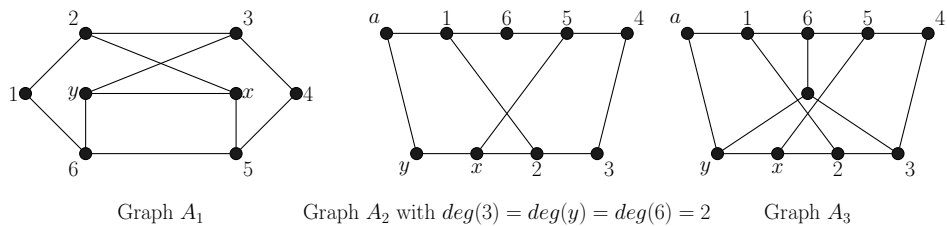
\begin{figure}[ht]
\begin{center}
\begin{tikzpicture}[scale=0.4, transform shape]
\tikzstyle{sommet}=[circle, draw,fill=black!90, inner sep=1pt, inner sep=4pt, minimum size=0.1cm]

\node[sommet] (x) at (-10,0) {};
\node[sommet] (2) at (-10,2) {};
\node[sommet] (6) at (-10,-2) {};
\node[sommet] (1) at (-12,0) {};
\node[sommet] (y) at (-5,0) {};
\node[sommet] (3) at (-5,2) {};
\node[sommet] (5) at (-5,-2) {};
\node[sommet] (4) at (-3,0) {};

\draw (6)--(1)--(2)--(3)--(4)--(5)--(6);
\draw (2)--(y)--(5); \draw (3)--(x)--(6);
\draw (x)--(y);

\node () at (-10.4,0) {\huge $y$};
\node () at (-10.4,2.5) {\huge $2$};
\node () at (-10.4,-2.5) {\huge $6$};
\node () at (-12.4,0) {\huge $1$};

\node () at (-4.6,0) {\huge $x$};
\node () at (-4.6,2.5) {\huge $3$};
\node () at (-4.6,-2.5) {\huge $5$};
\node () at (-2.6,0) {\huge $4$};

\node () at (-7.5,-4) {\huge Graph $A_1$};


\node[sommet] (a') at (0,2) {};
\node[sommet] (1') at (2,2) {};
\node[sommet] (2') at (4,2) {};
\node[sommet] (3') at (6,2) {};
\node[sommet] (4') at (8,2) {};

\node[sommet] (5') at (7,-2) {};
\node[sommet] (6') at (5,-2) {};
\node[sommet] (x') at (3,-2) {};
\node[sommet] (y') at (1,-2) {};

\draw (a')--(1')--(2')--(3')--(4')--(5')--(6')--(x')--(y')--(a');
\draw (1')--(6'); \draw (3')--(x');

\node () at (-0.4,2.5) {\huge $a$};
\node () at (2,2.5) {\huge $1$};
\node () at (4,2.5) {\huge $6$};
\node () at (6,2.5) {\huge $5$};
\node () at (8.4,2.5) {\huge $4$};

\node () at (7.4,-2.5) {\huge $3$};
\node () at (5,-2.5) {\huge $2$};
\node () at (3,-2.5) {\huge $x$};
\node () at (0.6,-2.5) {\huge $y$};

\node () at (4,-4) {\huge Graph $A_2$ with $deg(3)=deg(y)=deg(6)=2$};

\node[sommet] (a'') at (10,2) {};
\node[sommet] (1'') at (12,2) {};
\node[sommet] (2'') at (14,2) {};
\node[sommet] (3'') at (16,2) {};
\node[sommet] (4'') at (18,2) {};

\node[sommet] (5'') at (17,-2) {};
\node[sommet] (6'') at (15,-2) {};
\node[sommet] (x'') at (13,-2) {};
\node[sommet] (y'') at (11,-2) {};
\node[sommet] (c) at (14,0) {};

\draw (a'')--(1'')--(2'')--(3'')--(4'')--(5'')--(6'')--(x'')--(y'')--(a'');
\draw (1'')--(6''); \draw (3'')--(x'');
\draw (2'')--(c)--(y''); \draw (c)--(5'');

\node () at (9.6,2.5) {\huge $a$};
\node () at (12,2.5) {\huge $1$};
\node () at (14,2.5) {\huge $6$};
\node () at (16,2.5) {\huge $5$};
\node () at (18.4,2.5) {\huge $4$};

\node () at (17.4,-2.5) {\huge $3$};
\node () at (15,-2.5) {\huge $2$};
\node () at (13,-2.5) {\huge $x$};
\node () at (10.6,-2.5) {\huge $y$};

\node () at (14,-4) {\huge Graph $A_3$};
\end{tikzpicture}
\end{center}
\caption{Graphs $A_1$, $A_2$ and $A_3$}
\label{fig:subgraphs}
\end{figure}

It is not difficult to check that if $A_1$, $A_2$, or $A_3$ (see
Figure~\ref{fig:subgraphs}) is an induced subgraph of a subcubic graph, then we
can use Corollary \ref{lem:hred} as we have:

\begin{itemize}
\item $\alpha(A_1[\{2,3,5,6,x,y\}])=\alpha(A_1)=3$,
\item $\alpha(A_2[\{1,2,3,5,6,x,y\}])=\alpha(A_2)=4$,
\item $\alpha(A_3[\{1,2,3,5,6,x,y\}])=\alpha(A_3)=4$.
\end{itemize}

\begin{lemma}
If $A_1$, $A_2$, or $A_3$ is an induced subgraph of a subcubic graph $G$, then $\alpha(G-A_1)=\alpha(G)-3$,
$\alpha(G-A_2)=\alpha(G)-4$, $\alpha(G-A_3)=\alpha(G)-4$.
\end{lemma}

\begin{corollary}\label{lem:k23}

Let $G=(V,E)$ be a subcubic graph and $K\subseteq V$ such that $G[K]$ induces a
$K_{2,3}$. Then, if $G'$ is the graph obtained from $G$ by deleting the
vertices of $K$ and introducing a new vertex $z$ connected to $N(K)$, we have
(i) $\alpha(G')=\alpha(G)-2$ and (ii) if $G'$ contains an apple with a long
stem $A^*_p$, then $G$ also contains an apple with a long stem $A^*_{p'}$, with
$p'\ge p$.  \end{corollary}

\begin{proof}

To see that $\alpha(G')\ge \alpha(G)-2$ consider a maximum independent set of
$G$. If it contains at most two vertices from $K$ we are done, suppose then
that it contains three vertices. Then, it contains no vertices from $N(K)$. We
therefore augment this set with $z$ in $G'$. To see that $\alpha(G)\ge
\alpha(G')+2$ consider a maximum independent set of $G'$. If it is not using
$z$ then we augment it in $G$ by adding to it the vertices of the smaller part
of the $K_{2,3}$ induced by $K$ (these vertices have no neighbors outside $K$).
If it is using $z$, it is not using any vertices of $N(K)$, we therefore
replace $z$ by the three vertices of the larger part of the $K_{2,3}$.

For the second claim, suppose that $G'$ induces an $A^*_p$, for some $p$. If
this subgraph does not contain $z$, we are done, so suppose it does. If it
does, we find an apple with a long stem in $G$ by replacing $z$ with a vertex
of the smaller part of the $K_{2,3}$ in $G$, and also adding for each vertex of
$N(K)$ that belongs in the supposed $A^*_p$ one neighbor of that vertex from
$K$. It is not hard to see that the result is a sub-division of the original
$A^*_p$.  \end{proof}

\subsection{Applying graph reductions to large extended cycles}
\label{sec:analysis}

Let $G$ be an $(A^*_k,A^*_{k+1},\ldots)$-free subcubic graph. For ease of terminology and notation we will refer to any 
$A^*_t$ with $t\ge k$ simply as a large apple with a long stem. According to  Section~\ref{sec:to}, we may assume that $G$ contains a large 
extended cycle $C^*_p$, i.e. a graph that consists of an induced cycle of length $p$, plus two extra vertices which form a $C_6$ together with
four consecutive vertices of the cycle and have no other neighbours in $C^*_p$.
We denote the vertices of an extended cycle as shown in Figure~\ref{fig:extended}, where we have given labels to the vertices of the
$C_6$, plus some other interesting vertices. In the remainder we use simply $C^*$ to denote the extended cycle and $C_6$ to denote the set of
vertices $\{1,2,3,4,5,6\}$. Without loss of generality, we assume that $p\ge 3k$.

\begin{figure}[ht]
\centering
\begin{tikzpicture}[scale=0.3, transform shape]

\tikzstyle{vertex}=[circle, draw, inner sep=4pt, fill=black!100, minimum width=1.5pt, minimum size=0.15cm]

\draw (0,0) circle (5cm);
\node[vertex] (b) at (-4.97,0.7) {};
\node[vertex] (d) at (4.97,0.7) {};
\node[vertex] (a) at (-4.3,2.6) {};
\node[vertex] (c) at (4.3,2.6) {};
\node[vertex] (1) at (-3,4) {};
\node[vertex] (4) at (3,4) {};
\node[vertex] (2) at (-1.1,4.9) {};
\node[vertex] (3) at (1.1,4.9) {};
\node[vertex] (6) at (-1.1,2.5) {};
\node[vertex] (5) at (1.1,2.5) {};

\draw (1)--(6)--(5)--(4);


\node () at (-5.4,0.8) {\huge{\textbf{$b$}}};
\node () at (5.4,0.8) {\huge{\textbf{$d$}}};
\node () at (-4.8,2.6) {\huge{\textbf{$a$}}};
\node () at (4.8,2.6) {\huge{\textbf{$c$}}};
\node () at (-3.5,4.2) {\huge{\textbf{$1$}}};
\node () at (3.5,4.2) {\huge{\textbf{$4$}}};
\node () at (-1.1,5.4) {\huge{\textbf{$2$}}};
\node () at (1.1,5.4) {\huge{\textbf{$3$}}};
\node () at (-1.5,2.2) {\huge{\textbf{$6$}}};
\node () at (1.5,2.2) {\huge{\textbf{$5$}}};
\end{tikzpicture}
\caption{An extended cycle}
\label{fig:extended}
\end{figure}
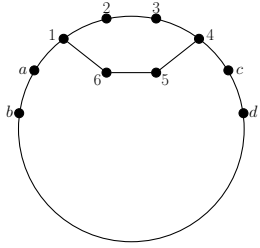

We will now go through a sequence of cases that covers all possible ways in
which $C^*$ may be connected to the rest of the graph.

\medskip
{\it Case 0}: Vertices $2$ and $3$ both have degree $2$ in $G$. In this case we
delete $2,3$ from the graph and add the edge connecting  $1$ to $4$. This decreases $\alpha(G)$ by exactly $1$.
Also, it is not difficult to check that this transformation 
does not create any new forbidden induced subgraphs. 

\medskip
Because of the above we can assume that the set $\{2,3\}$ has a neighbour
outside of $C^*$. We call this vertex $x$. Without loss of generality we assume
that $x$ is connected to $2$. Let us consider how $x$ is connected to the rest
of $C^*$.

\medskip
{\it Case 1.1}: $N(x)\cap C^*=\{2\}$. This case leads to a contradiction, as
$C^*\cup\{x\}\setminus\{3\}$ is a large apple with a long stem.

{\it Case 1.2}: $N(x)\cap C_6=\{2\}$ and $x$ has exactly one neighbour in
$C^*\setminus C_6$.  Let $f$ be that neighbour (which may coincide with one of
$a,b,c,d$).  Then the graph contains a large apple with a long stem: the stem
is made up of $\{5,6\}$, and the cycle from $2,x,f$, plus either the path from
$f$ to $1$, or the path from $f$ to $3$ in $C^*$ (whichever is longer).

{\it Case 1.3}: $N(x)\cap C_6=\{2\}$ and $x$ has two neighbours in $C^*\setminus
C_6$.  Let $f,g$ be these neighbours. If the distance from $f$ to $g$ in $C^*$
is at least $p/3$, then we have a large apple with a long stem: the cycle is
$x,f,g$ plus the path from $f$ to $g$ in $C^*$ and the stem is $2,3$.
Otherwise, one of $f,g$ has a path of length at least $p/3$ to $1$ or $4$ in
$C^*$ which does not contain the other vertex from $\{f,g\}$, so we find a
large apple with a long stem as in Case 1.2.

\medskip
From the above cases we conclude that $x$ has at least two neighbours in $C_6$.
Since the degrees of $1,4$ are already three in $C^*$, we conclude that $x$ has
at least two neighbours in $\{2,3,5,6\}$. 
Let us also rule out two further cases.

\medskip
{\it Case 1.4}: $N(x)\cap C_6=\{2,3\}$. If the degree of $x$ is 2, then we
can apply the $H$-subgraph reduction (Corollary\ref{lem:hred}) with $H=G[\{x\}]$,
which decrease $\alpha(G)$ by $1$. If $x$ has a neighbour outside of $C^*$,
we find a large apple with a long stem: $1654cd\ldots ba1$ and $2,x$. Finally,
if $x$ has a neighbour $f$ in $C^*\setminus C_6$, we find a large apple with a
long stem as in Case 1.3, where $\{5,6\}$ is the stem and the cycle is formed
by $2,x,f$, plus either the path from $f$ to $1$, or the path from $f$ to $3$
in $C^*$ (whichever is longer). 

{\it Case 1.5}: $|N(x)\cap C_6| = 3$. Here we can assume without loss of
generality that $N(x)\cap C_6=\{2,3,5\}$, as other cases are isomorphic to
this. Then, $\{2,3,4,5,x\}$ induces a house and we can apply the house-reduction.

\medskip

We are now ready to state the two main cases of our analysis:

\begin{lemma}\label{lem:x} 
If one of Cases 1.1-1.5 applies, then the instance can be simplified in
polynomial time. If none of Cases 1.1-1.5 applies, then either $N(x)\cap C_6 =
\{2,5\}$ or $N(x)\cap C_6=\{2,6\}$. 
\end{lemma}

We handle these two cases separately in the following subsections.

\subsubsection{$x$ is adjacent to 2 and 6}

\begin{lemma}\label{lem:typeA} Let $x$ be a vertex adjacent to $2$ and $6$ and
assume $x$ has a neighbour $y$ not in $C^*$. Then $G$ contains an induced
$\Phi$ or an induced $\Pi$ or an induced $\Gamma$ or an induced $\Theta$.
\end{lemma}

\begin{proof} 

If $y$ is adjacent to $3$, then by Lemma~\ref{lem:x} (and symmetry) $y$ is also
adjacent to $5$ and hence vertices $1,2,3,4,5,6,x,y$ induce a $\Theta$. 

If $y$ is adjacent to $c$, then vertices $2,3,4,x,y,c$ create a cycle of length
$6$ which, together with the path $1ab\ldots d$ gives a second large extended
cycle.  Therefore, by Lemma~\ref{lem:x} applied to this extended cycle, vertex
$5$ must be adjacent to $y$ and hence vertices $1,2,x,6,y,5,c,4$ induce a
$\Phi$.

If $y$ is adjacent to $a$, then vertices $a,y,1,2,x,6,3,4,5$ induce a $\Gamma$ with a possible missing common neighbour of $3$ and $5$
(any neighbour of these vertices must be common by Lemma~\ref{lem:x}).

If $y$ is adjacent to $b$ and not adjacent to $a$, then vertices
$a,b,y,1,2,x,6,3,4,5$ induce a $\Pi$ with a possible missing common neighbour
of $3$ and $5$ (any neighbour of these vertices must be common by
Lemma~\ref{lem:x}).

From now on, we assume $y$ has no neighbours in $\{3,5,a,b,c\}$. If $y$ has
neighbours on $C^*\setminus C_6$, then we can distinguish at most 3 cycles
containing $y$ as shown in Figure~\ref{fig:3cycles} (if $y$ has only 1
neighbour on $C^*\setminus C_6$, the cycle $C_2$ is missing).

\begin{figure}[ht]
\centering
\begin{tikzpicture}[scale=0.3, transform shape]

\tikzstyle{vertex}=[circle, draw, inner sep=4pt, fill=black!100, minimum width=1.5pt, minimum size=0.15cm]

\draw (0,0) circle (5cm);

\node[vertex] (2) at (0,5) {};

\node[vertex] (1) at (-2.5,4.3) {};
\node[vertex] (3) at (2.5,4.3) {};

\node[vertex] (a) at (-4.3,2.6) {};
\node[vertex] (c) at (4.3,2.6) {};
\node[vertex] (1') at (-4.3,-2.6) {};
\node[vertex] (3') at (4.3,-2.6) {};

\node[vertex] (y) at (0,0) {};
\node[vertex] (x) at (0,2.5) {};

\node[vertex] (b) at (-4.97,0.7) {};
\node[vertex] (d) at (4.97,0.7) {};

\draw (2)--(x)--(y)--(1');
\draw (y)--(3');


\node () at (-5.4,0.8) {\huge{\textbf{$b$}}};
\node () at (5.4,0.8) {\huge{\textbf{$c$}}};
\node () at (-4.8,2.6) {\huge{\textbf{$a$}}};
\node () at (4.8,2.6) {\huge{\textbf{$4$}}};

\node () at (0,5.6) {\huge{\textbf{$2$}}};

\node () at (-2.8,4.7) {\huge{\textbf{$1$}}};
\node () at (2.8,4.7) {\huge{\textbf{$3$}}};


\node () at (0.5,2.5) {\huge{\textbf{$x$}}};
\node () at (0.5,0.5) {\huge{\textbf{$y$}}};

\node () at (-2.5,1.5) {\Huge{\textbf{$C^1$}}};
\node () at (2.5,1.5) {\Huge{\textbf{$C^3$}}};
\node () at (0,-2.5) {\Huge{\textbf{$C^2$}}};

\end{tikzpicture}
\caption{Vertex $y$ has neighbours on $C$}
\label{fig:3cycles}
\end{figure}
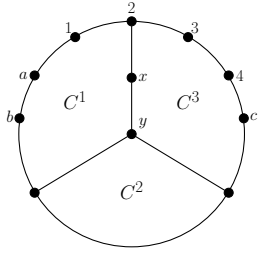

We observe that
at least one of the cycles $C^1$, $C^2$, $C^3$ is large, i.e. has length at least $p/3$.
Then $G$ contains a large apple with a long stem

\begin{itemize}
\item $C^*\cup\{x,y\}\setminus\{5,6\}$ if $y$ has no neighbours on $C^*\setminus C_6$,
\item $C^1\cup\{3,4\}$ if $C^1$ is large,
\item $C^2\cup\{x,2\}$ if $C^2$ is large,
\item $C^3\cup\{1,a\}$ if $C^3$ is large.
\end{itemize} 
A contradiction in all cases shows that $y$ has a neighbour in $\{3,5,a,b,c\}$ and hence $G$ contains an induced
$\Phi$ or an induced $\Pi$ or an induced $\Gamma$ or an induced $\Theta$
\end{proof}

We therefore find ourselves in the following context: $N(x)\cap C_6 = \{2,6\}$ and $N(x)\setminus C^*=\emptyset$. 
Before we proceed to the rest of the analysis, let us identify another relevant vertex. If $3$ has a
neighbour outside $C^*$ we call that vertex $y$. By Lemma \ref{lem:x} (and
appropriate symmetry) $y$ is also connected to $5$. We have also argued that
$x$ and $y$ are not adjacent. We will in the remainder assume that the degree
of $x$ is at least as large as the degree of $y$. This is without loss of
generality, as the two vertices can be exchanged by an appropriate automorphism
of $C^*$. In what follows, we analyze all possible adjacencies of $x$ and $y$ to the vertices of $C^*$.

\medskip
{\it Case 2.1}: If $x$ has degree 2 and $y$ does not exist (therefore $3,5$
have degree $2$), then we apply the $H$-subgraph reduction (Corollary~\ref{lem:hred}) with $H=\{x,3,5\}$, 
in which case $\alpha(G[H])=\alpha(G[N[H]])=3$ and hence  the removal of $N[H]$ decreases $\alpha(G)$  by $3$. 

{\it Case 2.2}: Assume $x$ has degree 2 and $y$ exists (therefore, $y$ is
connected to $3,5$).  We have assumed without loss of generality that $x$ has
at least as high degree as $y$, therefore $y$ has no other neighbour.  We
delete from the graph vertices $2,3,x,y$. If $G'$ is the new graph, we claim that
$\alpha(G')=\alpha(G)-2$. The inequality $\alpha(G')\ge\alpha(G)-2$ is clear,
since no independent set can take more than two of the deleted vertices. To see
that $\alpha(G)\ge\alpha(G')+2$, take a maximum independent set in $G'$. If it
contains vertex $5$, then it does not contain $4$ or $6$. Therefore, we can augment it
with $x,3$. If it contains $6$, we can augment it similarly by adding $y,2$.
Finally, if it contains neither $5$ nor $6$, we augment it with $x,y$.

{\it Case 2.3}: If $x$ is connected to $a$, $\{x,1,a,2,6\}$ induces a $K_{2,3}$, we can therefore invoke Corollary \ref{lem:k23} to simplify the graph.

{\it Case 2.4}: If $x$ is connected to $c$, then $x61ab\ldots cx$ together with
$3,4$ form a large apple with a long stem.

{\it Case 2.5}: If $x$ is connected to $d$, then $x21ab\ldots dx$ together with
$3,4$ form a large apple with a long stem.

{\it Case 2.6}: If $x$ is connected to a vertex $f$ of $C^*$ in the path from
$b$ to $d$ (but not $b$ or $d$), then: if $f$ is closer to $a$ than to $c$, we take
the path $xf\ldots dc432x$ plus $1,a$; otherwise we take $xf\ldots ba12x$ plus
$3,4$. In both cases these form a large apple with a long stem.

{\it Case 2.7}: If $x$ is connected to $b$ and $y$ does not exist, then we can
apply the $H$-subgraph reduction with $H=\{x,1,3,5\}$. It is not difficult to check that $\alpha(G[H])=\alpha(G[N[H]])=4$  
and hence the removal of $N[H]$ decreases $\alpha(G)$  by $4$. 

{\it Case 2.8}: Assume $x$ is connected to $b$, $y$ exists and it has degree 2
(that is, $y$ is connected only to $3,5$). We then delete from the graph the
vertices $\{x,y,1,2,3,5,6\}$ and add a new vertex $z$ adjacent to $a,b,4$. We
claim $\alpha(G')=\alpha(G)-3$. To see that $\alpha(G)\ge \alpha(G')+3$ take an
independent set of the new graph. If it does not include $z$ then we augment it
with $\{2,6,y\}$; if it does include $z$, it does not contain any of $a,b,4$,
so we replace $z$ with $\{1,x,3,y\}$. To see that $\alpha(G')\ge \alpha(G)-3$
take an independent set of $G$. If it contains at most three of the deleted
vertices we are done. If it contains four, these must be $\{1,x,3,5\}$,
therefore the set does not contain any of $a,b,4$; in this case we replace the
deleted vertices by $z$.

The new graph does not have a large apple with a long stem that uses $z$ and
both $a,b$, since that would induce a triangle. If, on the other hand, it has
an apple with a long stem that uses $z$ and at most two of its neighbors, then
$G$ also has a sub-divided copy of the same subgraph if we replace $z$ with
$1,2,3$.

{\it Case 2.9}: Finally, suppose $x$ is connected to $b$, $y$ exists and $y$
has degree $3$.  Since $x$ and $y$ have the same degree, we may exchange their
roles, and by symmetry and the same case analysis that we did for $x$ we
conclude that $y$ must be connected to $d$ (otherwise one of the previous cases
applies). We transform the graph as follows: we delete the vertices
$1,2,3,4,5,6,x,y$ and add two new vertices $z,w$ such that $z,w$ are connected
to each other, $z$ is connected to $a,b$, and $w$ is connected to $c,d$. We
claim that $\alpha(G')=\alpha(G)-3$.  First, to obtain $\alpha(G')\ge
\alpha(G)-3$, take a maximum independent set of $G$.  If it contains a vertex
from $a,b$ and a vertex from $c,d$, then it contains at most three of the
deleted vertices, since the six deleted vertices which are not adjacent to a
vertex of the independent set induce a $C_6$. In all other cases, the
independent set in $G$ contains at most four of the deleted vertices.  However,
if the set does not contain any of $a,b$, we can augment it with $z$ in $G'$,
while if it does not contain any of $c,d$ we can add to it $w$. To see that
$\alpha(G)\ge \alpha(G')+3$, take a maximum independent set in $G'$. If it is
using $z$, then it does not contain $a$ or $b$. In $G$ we replace $z$ with
$1,x,3,5$. The situation is symmetric if the set contains $w$.  Finally, if it
does not contain either $z$ or $w$, we observe that deleting the neighbours of
the set among the removed vertices gives a $C_6$, of which we can select three
vertices. The transformation does not introduce a new large apple with a long
set, since the closed neighbourhoods of $z,w$ include a triangle, therefore if
one or two of these vertices is used in the apple we can replace them with an
appropriate induced path through the deleted vertices in $G$.

\subsubsection{$x$ is adjacent to 2 and 5}

\begin{lemma}\label{lem:typeB} Let $x$ be a vertex adjacent to $2$ and $5$ and
assume $x$ has a neighbour $y$ not in $C^*$. Then $G$ contains an induced $A_1$
or an induced $A_2$ or an induced $A_3$ (Figure~\ref{fig:subgraphs}).
\end{lemma}

\begin{proof}
If $y$ is adjacent to $3$ or $6$, then $y$ is adjacent to both $3$ and $6$ (Lemma~\ref{lem:x}) and hence $G$ contains an induced $A_1$.

Assume $y$ is adjacent to $a$. Then, if all three vertices $3,6,y$ have degree
$2$ in $G$, then $G$ contains an induced $A_2$. If vertex $3$ has degree three,
it has a common neighbour with $6$ (by Lemma \ref{lem:x}), call this neighbour
$z$. We claim that $z$ must also be connected to $y$, which will give an
induced $A_3$. To see this, consider the set of vertices
$(C^*\setminus\{2,3\})\cup\{x,y\}$. This set induces an extended cycle, where
the $C_6$ is now formed by $a,1,6,5,x,y$. Since $z$ is connected to $6$, it
must be connected to one of $\{x,y\}$ (Lemma \ref{lem:x}). However, $x$ already
has three neighbours ($2,5,y$), therefore, $z$ is connected to $y$.

If $y$ is adjacent to $c$ this is symmetric to $y$ being adjacent to $a$. So,
we suppose that $y$ is adjacent to none of $3,6,a,c$.  The rest of the proof is
similar to that of Lemma~\ref{lem:typeA} with the only difference that if $y$
is adjacent only to $b$ this time we can find a large apple with a long stem,
where the stem is $\{1,6\}$ and the cycle goes through $byx234cd\ldots b$.
\end{proof}

\begin{lemma}\label{lem:typeB-0} Let $x$ be a vertex adjacent to $2$ and $5$
and assume $x$ has a neighbour in $C^*\setminus C_6$. Then this neighbour is
one of $a$ and $c$.  \end{lemma}

\begin{proof}

Let $f$ be the neighbour of $x$ in $C^*\setminus C_6$ (note that $f$ may
coincide with $b$ or $d$). Suppose that the distance in $C^*\setminus C$ from
$f$ to $a$ is at least as large as the distance from $f$ to $c$ (the other case
is symmetric). We take the cycle $xf\ldots a12x$ and the stem $\{3,4\}$ to form a large apple with a long stem.
\end{proof}

Because of symmetry, we will in the remainder assume without loss of generality
that if $x$ has a neighbour in $C^*\setminus C_6$, then that neighbour is $a$. We
recall that, since $x$ is connected to $2,5$, if $3$ has a neighbour outside of
$C^*$, this neighbour is common with $6$. We will call such a vertex (if it
exists) $y$. By the same reasoning that we applied for $x$, vertex $y$ cannot have a
neighbour outside $C^*$ (therefore $x$ and $y$ are not adjacent), and if it has
a neighbour in $C^*\setminus C_6$, this must be $c$. As in the previous section,
we will assume without loss of generality that the degree of $x$ is at least as
high as that of $y$, otherwise we can exchange their roles.

\medskip
{\it Case 3.1}: If $x$ has degree two and $y$ does not exist, then we can apply
the $H$-subgraph reduction (Corollary~\ref{lem:hred}) with $H=\{x,3,6\}$, in
which case $\alpha(G[H])=\alpha(G[N[H]])=3$ and hence the removal of $N[H]$
decreases $\alpha(G)$ by $3$.

{\it Case 3.2}: Assume $x$ has degree two and $y$ exists: $y$ is adjacent to $3,6$
and no other vertex, since we assumed that $x$ has degree at least as high as
$y$. We now remove from the graph the vertices $\{x,y,1,2,3,4,5,6\}$ and add a
new vertex $z$ adjacent to $a,c$. We claim that $\alpha(G')=\alpha(G)-3$. To
see that $\alpha(G')\ge \alpha(G)-3$ take an independent set of $G$. If it
contains at most three of the deleted vertices, we are done. If it contains
four, then it must contain both $x$ and $y$, which implies that it contains $1$
and $4$.  The set therefore does not contain $a$ or $c$, so the deleted
vertices can be replaced by $z$. For the other direction, to see that
$\alpha(G)\ge \alpha(G')+3$, take an independent set of $G'$. If it does not
contain $z$, we augment the set with $x,3,6$; if it does, then it does not
contain $a$ or $c$, so we replace $z$ with $\{1,4,x,y\}$.  Our transformation
does not introduce a new forbidden induced subgraph, as any path through $z$ in
the transformed graph can be mapped to the path $a1234c$ in $G$.

{\it Case 3.3}: Assume $x$ is adjacent to $a$ and $y$ does not exist. In this case we
delete $x$ from the graph and claim that the independence number is unchanged.
To see this, take a maximum independent set $S$ in $G$. If $x\not\in S$ we are
done.  Suppose then that $x\in S$, therefore $S$ does not contain any of
$a,2,5$.  As a result, it contains at most two vertices from $C_6$. Consider
now the set $(S\setminus(C_6\cup\{x\}))\cup\{1,3,5\}$. This is a valid
independent set (since $S$ does not contain $a$) of the same size as $S$.

{\it Case 3.4}: Assume $x$ is adjacent to $a$, $y$ exists and it has degree $2$. In
this case we delete from the graph vertex $6$. We claim that the independence
number stays unchanged. To see this, take a maximum independent set $S$. If
$6\not \in S$ we are done, so suppose that $6\in S$, therefore $1,5,y\not\in
S$. If $3\not\in S$, then $S\setminus\{6\}\cup\{y\}$ is an independent set of
the same size in the new graph, and we are done. Suppose then that $3\in S$,
which means that $2,4\not\in S$. We now observe that the set
$S\setminus\{a,x,3,6\}\cup\{2,5,y\}$ is an independent set of size $|S|$ in the
new graph. To see that it has the same size, we note that $a$ is adjacent to
$x$. To see that it is independent, we note that $S\setminus\{a,x,3,6\}$ does
not contain any neighbours of $\{2,5,y\}$.

{\it Case 3.5}: If $x$ is adjacent to $a$ and $y$ exists and is adjacent to
$c$, then we can apply the $H$-subgraph reduction with $H=\{x,y,1,4\}$, in
which case $\alpha(G[H])=\alpha(G[N[H]])=4$ and hence the removal of $N[H]$
decreases $\alpha(G)$ by 4.

\section{Conclusion}

Summarizing the discussion in the previous sections, we make the following conclusion, which extends several previously known results.  

\begin{theorem}
The maximum independent set problem can be solved in polynomial time in the class of  $(A^*_k,A^*_{k+1},\ldots)$-free subcubic graphs
for any fixed value of $k$.
\end{theorem}

Since $A^*_t$ contains $S_{2,k,k}$ for any $t>k$, we derive the following corollary from this theorem. 
\begin{corollary}
The maximum independent set problem can be solved in polynomial time in the class of  $S_{2,k,k}$-free subcubic graphs
for any fixed value of $k$.
\end{corollary} 

This result brings us much closer to the dichotomy described in Conjecture~\ref{con:1}. However, proving this conjecture in its whole generality remains a challenging open problem.

\bibliographystyle{plain}
\bibliography{s2kk-full}

\end{document}